\documentclass[english,11pt,a4paper]{article}
\usepackage[T1]{fontenc}
\usepackage[utf8]{inputenc}
\usepackage{geometry}
\usepackage{babel}
\usepackage{hyperref}
\usepackage{enumitem}

\usepackage{amsmath, amssymb, amstext, mathtools}
\usepackage{amsthm}
\usepackage{complexity}

\usepackage{tikz}
\usepackage{pgfplots}
\usepackage{textpos}

\usepackage[capitalize]{cleveref}

\definecolor[named]{urlblue}{cmyk}{1,0.58,0,0.21}

\hypersetup{
    breaklinks=true,
    colorlinks=true,
    citecolor=purple!70!blue!60!black,
    linkcolor=purple!70!blue!90!black,
    urlcolor=urlblue,
    pdflang={en},
    pdftitle={Treedepth Inapproximability and Exponential ETH Lower Bound},
    pdfauthor={\'{E}douard Bonnet, Daniel Neuen and Marek Soko{\l}owski}
}

\newtheorem{theorem}{Theorem}[section]
\newtheorem{lemma}[theorem]{Lemma}
\newtheorem{corollary}[theorem]{Corollary}

\newtheorem{observation}[theorem]{Observation}

\newtheorem{claim}[theorem]{Claim}

\theoremstyle{definition}

\theoremstyle{remark}

\newenvironment{proofofclaim}{\begin{proof}[Proof of Claim]}{\end{proof}}

\renewcommand{\geq}{\geqslant}
\renewcommand{\leq}{\leqslant}

\newcommand\conncomp{\text{cc}}

\DeclareMathOperator{\tw}{tw}
\DeclareMathOperator{\td}{td}
\DeclareMathOperator{\vc}{vc}

\newcommand{\orcid}[1]{\href{https://orcid.org/#1}{\includegraphics[height=1.8ex]{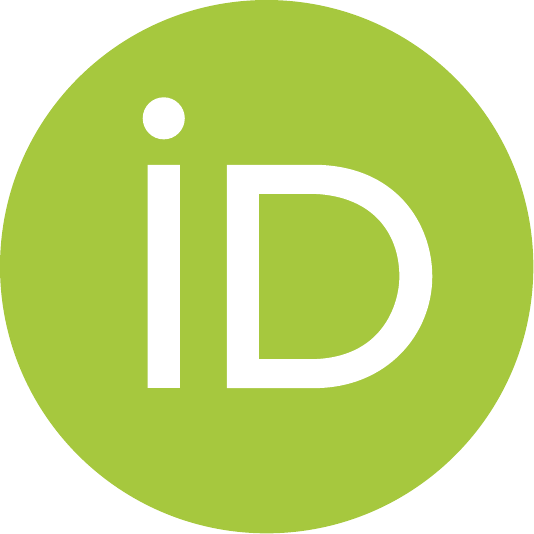}}}

\title{Treedepth Inapproximability and \\Exponential ETH Lower Bound}
\author{\'{E}douard Bonnet \orcid{0000-0002-1653-5822}\\
  CNRS, ENS de Lyon, Université Claude Bernard Lyon 1, LIP UMR 5668, France
\and
Daniel Neuen \orcid{0000-0002-4940-0318}\\
Max Planck Institute for Informatics, Saarland Informatics Campus, Germany
\and
Marek Soko{\l}owski \orcid{0000-0001-8309-0141}\\
Max Planck Institute for Informatics, Saarland Informatics Campus, Germany}

\date{}

\begin{document}

\maketitle

\begin{abstract}
  Treedepth is a~central parameter to algorithmic graph theory.
  The current state-of-the-art in computing and approximating treedepth consists of a~$2^{O(k^2)} n$-time exact algorithm and a~polynomial-time $O(\text{OPT} \log^{3/2} \text{OPT})$-approximation algorithm, where the former algorithm returns an \emph{elimination forest} of height~$k$ (witnessing that treedepth is at~most~$k$) for the $n$-vertex input graph $G$, or correctly reports that $G$ has treedepth larger than~$k$, and $\text{OPT}$ is the actual value of the treedepth.
  On the complexity side, exactly computing treedepth is \textsf{NP}-complete, but the known reductions do not rule out a~polynomial-time approximation scheme (PTAS), and under the Exponential Time Hypothesis (ETH) only exclude a~running time of~$2^{o(\sqrt n)}$ for exact algorithms.

  We show that 1.0003-approximating \textsc{Treedepth} is \textsf{NP}-hard, and that exactly computing the treedepth of an $n$-vertex graph requires time $2^{\Omega(n)}$, unless the ETH fails.
  We further derive that there exist absolute constants $\delta, c > 0$ such that any $(1+\delta)$-approximation algorithm requires time $2^{\Omega(n / \log^c n)}$.
  We do so via a~simple direct reduction from \textsc{Satisfiability} to \textsc{Treedepth}, inspired by a~reduction recently designed for \textsc{Treewidth} [STOC~'25].
\end{abstract}

\section{Introduction}

The treedepth $\td(G)$ of a~graph $G$ is the least integer~$k$ such that there is a~rooted forest~$F$ of height~$k$ with same vertex set as~$G$ such that every edge of~$G$ is between two nodes in ancestor--descendant relationship in~$F$.
Treedepth and treewidth, $\tw$, are related by the inequalities $\tw(G)+1 \leqslant \td(G) \leqslant \tw(G) \cdot (1 + \log n)$, for every $n$-vertex graph $G$.
An $n$-vertex path has treewidth (even pathwidth)~1, but treedepth $\Theta(\log n)$.

Treedepth comes into play in various contexts.
Notably, in the sparsity theory initiated by Ne\v set\v ril and Ossona de Mendez~\cite{NesetrilM12}, treedepth provides a~characterization of classes of bounded expansion (roughly speaking, classes excluding, as subgraphs, short subdivisions of graphs of large average degree).
Graph classes with bounded expansion are exactly those with so-called \emph{low treedepth covers} (some form of cover by graphs of bounded treedepth)~\cite{Nesetril08}.

Graphs of bounded treewidth famously lend themselves to fixed-parameter tractable (FPT) algorithms for various \mbox{\textsf{NP}-hard} problems, with parameter the width of the computed (or given) tree-decompositions, by performing dynamic programming over these decompositions (see, e.g., \cite[Chapter 7]{CyganFKLMPPS15}).
However, this method consumes essentially as much space as it takes time; in particular, most of these algorithms for \mbox{\textsf{NP}-hard} problems take exponential space in the treewidth.
Bounded treedepth, in contrast, often allows for parameterized algorithms with comparable running time but using only polynomial space; see~for instance \cite{Furer17,NederlofPSW23,PilipczukS21,PilipczukW18}.
Other uses of treedepth can be found in formula complexity~\cite{Kush23}, distributed model checking~\cite{Fomin24}, product structure theory~\cite{Dujmovic24}, relation to polynomial minors \cite{KawarabayashiR21,Hatzel24}, graph circumference~\cite{BrianskiJMMSS23} etc.
We now survey the current state of the art on computing and approximating the treedepth of an input graph---the topic of the current paper.
Note that \textsc{Treedepth} was the selected problem for the 2020 edition of the PACE challenge~\cite{Kowalik20}.

The decision version of \textsc{Treedepth} is \textsf{NP}-complete~\cite{Pothen88,BodlaenderDJKKMT98}.
There is an easy algorithm that computes the treedepth of an $n$-vertex graph in time $O(2^n \cdot n)$, and a~slightly faster exponential algorithm computing the decomposition in time $O(1.9602^n)$~\cite{FominGP15}.
An FPT (exact) algorithm in $2^{O(k^2)} n$~time has been established, first needing exponential space~\cite{ReidlRVS14}, and later improved to only use polynomial space~\cite{NadaraPS22}.
More precisely, these algorithms run in time $2^{O(\td(G) \tw(G))} n$ on an $n$-vertex graph $G$.
An outstanding open question is whether running time $2^{O(k)} n^{O(1)}$ can be obtained.
Notably, \textsc{Treewidth} admits constant-approximation algorithms in this running time.
In contrast to treewidth, such a~result remains elusive for \textsc{Treedepth}.
On the front of polynomial-time approximation algorithms, the best factor that currently can be achieved is $O(\tw(G) \log^{3/2} \tw(G))$~\cite{CzerwinskiNP19}, hence $O(\td(G) \log^{3/2} \td(G))$.

The previously known hardness results~\cite{Pothen88,BodlaenderDJKKMT98} are rather unsatisfactory:
First of all, they did not rule out a~polynomial-time approximation scheme (PTAS) for \textsc{Treedepth}.
Furthermore, by following the reductions presented in both papers, we can only infer a~$2^{\Omega(\sqrt{n})}$ lower bound for the exact variant of the problem under the Exponential Time Hypothesis (ETH).\footnote{The assumption that there is some $\lambda > 1$ such that $n$-variable \textsc{3-SAT} requires time $\Omega(\lambda^n)$.}
Naturally, this excludes any $2^{o(\sqrt{k})} n^{O(1)}$-time parameterized exact algorithm for the problem under ETH.

\subparagraph{Our Results.}
In this work, we design a~simple linear reduction from \textsc{Satisfiability} to \textsc{Treedepth}.
It draws inspiration from a~recent similar result for \textsc{Treewidth} \cite{tw-hardness}, and also relies on the hardness of \textsc{Vertex Cover} (i.e., the task of finding a~smallest vertex subset hitting every edge) on tripartite graphs.
Our reduction yields the following inapproximability results and computational lower bounds for \textsc{Treedepth}.

\begin{theorem}
  \label{thm:td-no-ptas}
  It is \textsf{NP}-hard to 1.0003-approximate \textsc{Treedepth}.
\end{theorem}

In particular, the theorem rules out a~polynomial-time approximation scheme (PTAS) for \textsc{Treedepth} (assuming $\textsf{P} \neq \textsf{NP}$).

\begin{theorem}
  \label{thm:td-eth-lower-bound}
  Assuming ETH, there is some $\varepsilon > 0$ such that the treedepth of an~$n$-vertex graph cannot be computed in time $O(2^{\varepsilon n})$.
\end{theorem}

This also excludes $2^{o(k)}n^{O(1)}$-time parameterized exact algorithms for the problem under ETH.

In fact, we obtain that even approximating treedepth to a small constant factor requires almost exponential time.

\begin{theorem}
  \label{thm:td-eth-approx-lower-bound}
  Assuming ETH, there exist absolute constants $\delta, \varepsilon, c > 0$ such that $(1 + \delta)$-approximating the treedepth of an~$n$-vertex graph cannot be done in time $O(2^{ \varepsilon n / \log^c n})$.
\end{theorem}

\section{Preliminaries}

In this work, we consider only simple undirected graphs with no self-loops.
For a~graph~$G$, we define $V(G)$ to be the set of vertices, $E(G)$ to be the set of edges, and $\conncomp(G)$ to be the set of connected components of $G$.
Also, we write $N_G(v)$ for the set of neighbors of a vertex $v \in V(G)$ in the graph $G$.
A~forest $F$ is a~graph without cycles; it is rooted if each connected component of $F$ (a~\emph{tree}) has a~\emph{root}.
The depth of a~forest is then the number of vertices on the longest root-to-leaf path in $F$.
A~vertex $u$ is an~\emph{ancestor} of $v$ in $F$ if $u$ lies on the path between $v$ and the root of the tree of $F$ containing $v$; we equivalently then say that $v$ is a~\emph{descendant} of $u$.
In particular, every vertex is its own ancestor and descendant.

A~set $S \subseteq V(G)$ is a~\emph{vertex cover} if each edge of $G$ is incident to at least one vertex of $S$.
The \emph{vertex cover number} of $G$, denoted by $\vc(G)$, is the minimum cardinality of a~vertex cover of $G$.

The treedepth of $G$, denoted $\td(G)$, is defined recursively as follows:
\[
  \td(G) = \begin{cases}
    0 & \text{if $G$ has no vertices,} \\
    \max_{H \in \conncomp(G)} \td(H) & \text{if $G$ is disconnected,} \\
    \min_{v \in V(G)} 1 + \td(G - v) & \text{if $G$ is connected.}
  \end{cases}
\]
We also use an~equivalent definition of treedepth involving \emph{elimination forests}.
Here, we say that a~rooted forest $F$ is an~elimination forest of $G$ if $V(F) = V(G)$ and for every edge $uv$ of $G$, vertices $u$ and $v$ are in the ancestor--descendant relationship in $F$.
Then, the treedepth of $G$ is the minimum possible depth of an~elimination forest of $G$.

We use the following straightforward facts in our arguments:
\begin{observation}
  \label{obs:td-clique}
  If $K$ is a~clique in $G$, then in every elimination forest of $G$, all vertices of $K$ are contained in a~single root-to-leaf path.
\end{observation}
\begin{observation}
  \label{obs:td-conn-root}
  If for some nonempty $S \subseteq V(G)$, the induced subgraph $G[S]$ is connected, then in every elimination forest of $G$ some vertex of $S$ is an~ancestor of all elements of $S$.
\end{observation}

We say that a~set $X \subseteq V(G)$ of vertices of $G$ is a~\emph{set of (false) twins} if $N_G(v) = N_G(w)$ for all $v,w \in X$.
Observe that $X$ is an independent set in this case.
Every inclusion-wise minimal vertex cover of $G$ either contains $X$ in its entirety or is disjoint from $X$.

In this paper, we work with formulas with Boolean variables, say $x_1, \ldots, x_n$.
A~\emph{literal} is a~formula of the form $x_i$ or $\neg x_i$.
For an~integer $k$, a~Boolean formula is said to be in \emph{$k$-CNF form} (or: is a~\emph{$k$-CNF formula}) if it is a~conjunction of \emph{clauses}: subformulas of the form $\ell_1 \vee \ldots \vee \ell_k$ for literals $\ell_1, \ldots, \ell_k$, each containing a~different variable of the formula.

\section{Treedepth and Vertex Cover of Tripartite Graphs}

We say that a~graph $G = (V, E)$ is tripartite if there exists a~tripartition $V = A \cup B \cup C$ such that the subgraphs of $G$ induced by $A$, $B$, and $C$, respectively, are edgeless.
We argue that it is possible to extend each tripartite graph $G$ to a~supergraph $H$ by adding suitable clique gadgets so that the treedepth of $H$ is tightly controlled by the vertex cover number of~$G$.

\begin{lemma}
  \label{lem:vc-to-td}
  Let $G = (V, E)$ be a~tripartite graph with tripartition $V = A \cup B \cup C$ such that $A, B, C$ are nonempty and let $\ell$ be a~positive integer such that $\ell \geq \vc(G)$.
  Consider a~supergraph $H = (V', E')$ of $G$ created by adding to $G$ three $\ell$-vertex cliques $K_A$, $K_B$, $K_C$ and three additional vertices $z_A$, $z_B$, $z_C$, and adding for each $X \in \{A, B, C\}$ all edges between $K_X$ and $X \cup \{z_X\}$.
  That is,
  \[ \begin{split}
    V' & = V \cup K_A \cup K_B \cup K_C \cup \{z_A, z_B, z_C\}, \\
    E' & = E\ \cup\ (A \cup \{z_A\}) \times K_A\ \cup\ (B \cup \{z_B\}) \times K_B\ \cup\ (C \cup \{z_C\}) \times K_C.
  \end{split} \]
  Then $\td(H) = \vc(G) + \ell + 1$.
\end{lemma}
\begin{proof}
  Let $S$ be the minimum vertex cover of $G$.
  Then $H - S$ has three connected components, namely, for each $X \in \{A, B, C\}$, the subgraph of $H$ induced by $K_X \cup \{z_X\} \cup (X \setminus S)$.
  Noting that $|K_X| = \ell$ and $\{z_X\} \cup (X \setminus S)$ is an~independent set in $H$, we have $\td(H[K_X \cup \{z_X\} \cup (X \setminus S)]) \leq |K_X| + 1 = \ell + 1$, and therefore
  \[
    \td(H) \leq |S| + \max_{X \in \{A, B, C\}} \td(H[K_X \cup \{z_X\} \cup (X \setminus S)]) \leq \vc(G) + \ell + 1.
  \]

  We now move to the lower bound on $\td(H)$.
  Aiming for contradiction, suppose that $\td(H) \leq \vc(G) + \ell$; then there exists an~elimination forest $F$ of $H$ of depth at most $\vc(G) + \ell$.

  By \cref{obs:td-clique}, for each $X \in \{A, B, C\}$, all vertices of $K_X$ belong to a~single root-to-leaf path (possibly different for different choices of $X$).
  We can thus choose three vertices $\kappa_A \in K_A$, $\kappa_B \in K_B$, $\kappa_C \in K_C$ as the deepest vertices of the respective cliques in $F$.

  \begin{claim}
    \label{cl:depth-lower-bound-gadget}
    Let $X \in \{A, B, C\}$.
    Suppose that $S \subseteq V' \setminus (K_X \cup \{z_X\})$ is so that all vertices of $S$ are ancestors of $\kappa_X$ in $F$.
    Then the depth of $F$ is at least $|S| + \ell + 1$.
  \end{claim}
  \begin{proofofclaim}
    All vertices of $K_X$ are ancestors of $\kappa_X$, and $z_X$ is either an~ancestor or a~descendant of $\kappa_X$.
    In either case, vertices of $S \cup K_X \cup \{z_X\}$ lie on a~single root-to-leaf path in $F$, implying that the depth of $F$ is at least $|S| + \ell + 1$.
  \end{proofofclaim}

  \begin{claim}
    \label{cl:no-deepest-vertices}
    No pair of vertices in $\{\kappa_A, \kappa_B, \kappa_C\}$ is in ancestor--descendant relationship in $F$.
  \end{claim}
  \begin{proofofclaim}
    Suppose without loss of generality that $\kappa_A$ is an~ancestor of $\kappa_B$.
    Then all vertices of $K_A$ are ancestors of $\kappa_B$, so from \cref{cl:depth-lower-bound-gadget} we infer that the depth of $F$ is at least $|K_A| + \ell + 1 = 2\ell + 1 > \vc(G) + \ell$; a~contradiction.
  \end{proofofclaim}

  \begin{claim}
    \label{cl:edge-ancestor}
    Let $X, Y \in \{A, B, C\}$ with $X \neq Y$ and pick $v_X \in X$, $v_Y \in Y$ connected by an~edge.
    Then one of the vertices $v_X, v_Y$ is an~ancestor of both $\kappa_X$ and $\kappa_Y$.
  \end{claim}
  \begin{proofofclaim}
    Let $S = \{v_X, v_Y, \kappa_X, \kappa_Y\}$.
    Noting that $\kappa_X v_X v_Y \kappa_Y$ is a~path in $H$, we have by \cref{obs:td-conn-root} that some vertex of $S$ is an~ancestor of all vertices in $S$.
    Since $\kappa_X$ and $\kappa_Y$ are not in the ancestor--descendant relationship (\cref{cl:no-deepest-vertices}), the claim follows.
  \end{proofofclaim}

  Let us now resolve a~degenerate case where at least one of the sides of $G$ (say, $C$) is not incident to any edge in $G$.
  Define $Q$ to be the set of vertices of $A \cup B$ that are ancestors of both $\kappa_A$ and $\kappa_B$; by \cref{cl:edge-ancestor}, $Q$ is a~vertex cover of $G$.
  Hence by \Cref{cl:depth-lower-bound-gadget} for $X = A$, the depth of $F$ is at least $|Q| + \ell + 1 > \vc(G) + \ell$; a~contradiction.

  Thus each side of $G$ is incident to an~edge of $G$.
  Therefore, we can easily verify that $H$ is connected and so $F$ is a~single rooted tree.
  Let then $u_{AB}$ be the lowest common ancestor of $\kappa_A$ and $\kappa_B$ in $F$; analogously define $u_{BC}$ and $u_{CA}$.
  The three newly defined vertices pairwise remain in the ancestor--descendant relationship: for instance, $u_{AB}$ and $u_{BC}$ are both ancestors of $\kappa_B$, so one of them is an~ancestor of the other.
  In particular, $u_{AB}$, $u_{BC}$ and $u_{CA}$ lie on a~single root-to-leaf path, and (at least) one of these vertices---say, $u_{AB}$---is a~descendant of all three.
  Let $Q \subseteq V$ be the set of vertices of $G$ that are ancestors of $\kappa_A$.
  Applying \cref{cl:edge-ancestor} multiple times, we observe that:
  \begin{itemize}
    \item for each $v_Av_B \in E(G)$ with $v_A \in A$, $v_B \in B$, either $v_A \in Q$ or $v_B \in Q$;
    \item for each $v_Av_C \in E(G)$ with $v_A \in A$, $v_C \in C$, either $v_A \in Q$ or $v_C \in Q$;
    \item for each $v_Bv_C \in E(G)$ with $v_B \in B$, $v_C \in C$, either $v_B$ or $v_C$ is an~ancestor of both $\kappa_B$ and $\kappa_C$ in $F$; hence it is also an~ancestor of $u_{BC}$, so also an~ancestor of $u_{AB}$ and thus also $\kappa_A$. Consequently either $v_B \in Q$ or $v_C \in Q$.
  \end{itemize}

  We conclude that $Q$ is a~vertex cover of $G$; and so by \Cref{cl:depth-lower-bound-gadget}, the depth of $F$ is at least $|Q| + \ell + 1 > \vc(G) + \ell$; a~contradiction.
\end{proof}

\section{Hardness Proofs}

In this section, we will prove the announced hardness results (\cref{thm:td-no-ptas,thm:td-eth-lower-bound,thm:td-eth-approx-lower-bound}).
All reductions will start from an~instance $\varphi$ of {\sc Satisfiability} in $k$-CNF form, in which every variable occurs a~bounded number of times, and produce a~graph (or a~family of graphs) whose treedepth tightly depends on the maximum number of clauses that can be satisfied in $\varphi$.
At the heart of our framework lies a~construction of a~tripartite graph adapted from the work of Bonnet~\cite{tw-hardness}, which we formally describe below.

Let $\varphi$ be a~$k$-CNF formula and $p$ be an~integer.
We then define a~tripartite graph $G(\varphi, p)$ as follows.
Let $\gamma = 2^{k-1} p$.
Define the vertex set of $G(\varphi, p)$ as $A \cup B_+ \cup B_-$, where:
\begin{itemize}
  \item For every clause $C_i = \ell_1 \vee \ldots \vee \ell_k$ of $\varphi$ and every possible choice $s_1 \in \{\ell_1, \neg \ell_1\}, \ldots, s_k \in \{\ell_k, \neg \ell_k\}$ such that $(s_1, \ldots, s_k) \neq (\neg \ell_1, \ldots, \neg \ell_k)$, we add to $A$ a~vertex $a_i(s_1, \ldots, s_k)$. Let $A(C_i)$ be the set of all vertices added to $A$ for clause $C_i$. Intuitively, $A(C_i)$ contains all valuations of variables represented by the literals $\ell_1, \ldots, \ell_k$ that satisfy the clause $C_i$.
  \item For every variable $x_j$ of $\varphi$ and every $t \in [\gamma]$, we add to $B_+$ a~vertex $b_{j,t}$ and to $B_-$ a~vertex $c_{j,t}$.
    Let then $B_+(x_j) = \{b_{j,1}, \ldots, b_{j,\gamma}\}$ and $B_-(x_j) = \{c_{j,1}, \ldots, c_{j,\gamma}\}$.
    Also, for convenience, let $B(x_j) = B_+(x_j)$ and $B(\neg x_j) = B_-(x_j)$.
\end{itemize}

We will now construct the set of edges of $G(\varphi, p)$ to ensure that every valuation of variables of $\varphi$ corresponds to an~inclusion-wise minimal vertex cover $S$ that includes: ($i$) all vertices of $A$, except one vertex of $A(C_i)$ for each satisfied clause $C_i$, corresponding to the valuation of the variables represented by the literals of $C_i$, and ($ii$) all vertices of $B_+(x_j)$ if $x_j$ is evaluated positively, or $B_-(x_j)$ if $x_j$ is evaluated negatively.
To this end, we construct the set $E$ of edges of $G(\varphi, p)$ via the following process:
\begin{itemize}
  \item Add every edge between $B(x_j)$ and $B(\neg x_j)$ for every variable $x_j$;
  \item Connect each vertex $a_i(s_1, \ldots, s_k) \in A$ with every vertex of $B(s_1) \cup \ldots \cup B(s_k)$.
\end{itemize}

This concludes the construction.
We now show that the vertex cover number of $G(\varphi, p)$ is controlled by the maximum number of clauses satisfiable by~$\varphi$.

\begin{lemma}
  \label{lem:formula-to-vc}
  Let $\varphi$ be a~$k$-CNF formula with $n$ variables and $m$ clauses where every variable appears at most $2p + 1$ times.
  Suppose also that $m'$ is the maximum number of clauses that can be satisfied in $\varphi$.
  Then \[\vc(G(\varphi, p)) = (2^k - 1)m - m' + 2^{k-1} pn.\]
\end{lemma}
\begin{proof}
  First, suppose that $F$ is a~valuation of variables of $\varphi$ that satisfies $m'$ clauses; we think of $F$ as a~set of literals that includes exactly one literal from $\{x_j, \neg x_j\}$ for each variable~$x_j$.
  We define a~set $S$ of vertices of $G(\varphi, p)$ by including:
  \begin{itemize}
    \item every vertex of the form $a_i(s_1, \ldots, s_k) \in A$ such that $\{s_1, \ldots, s_k\} \not\subseteq F$; and
    \item all vertices of $B(\ell_j)$ for every literal $\ell_j \in F$.
  \end{itemize}

  Observe that $S$ is a~vertex cover of $G(\varphi, p)$.
  Indeed, every edge between $B_+$ and $B_-$ is covered by $S$.
  Next, consider the edges between $A$ and $B:= B_+ \cup B_-$.
  Pick a~vertex $a_i(s_1, \ldots, s_k) \in A \setminus S$.
  Every neighbor of this vertex belongs to $B(s_1) \cup \ldots \cup B(s_k)$.
  By~construction of $S \cap A$, we have $\{s_1, \ldots, s_k\} \subseteq F$, and so $B(s_1) \cup \ldots \cup B(s_k) \subseteq S$.

  Now let us determine the size of $S$.
  Since $F$ satisfies $m'$ clauses of $\varphi$, there exist exactly $m'$ vertices of the form $a_i(s_1, \ldots, s_k) \in A$ such that $s_1, \ldots, s_k \in F$.
  These are precisely the vertices of $A$ that do not belong to $S$.
  Therefore $|A \setminus S| = m'$ and so $|S \cap A| = (2^k - 1)m - m'$.
  Also, $|S \cap B| = \gamma n = 2^{k-1} p n$.
  Therefore, $|S| = (2^k - 1)m - m' + 2^{k-1}pn$.

  \smallskip

  On the other hand, suppose that $S$ is a~minimum-cardinality vertex cover of $G(\varphi, p)$; and out of those, $S$ contains the fewest number of vertices of~$B$.
  Assume that $|S| \leq (2^k - 1)m - m'' + 2^{k-1}pn$, aiming to show that at least $m''$ clauses of $\varphi$ can be satisfied.

  For every variable $x_j$ of $\varphi$, each of the sets $B_+(x_j)$ and $B_-(x_j)$ is a set of twins in $G(\varphi, p)$, so $S$ either contains it fully or is disjoint from it; and moreover, $S$ must contain at least one of the sets $B_+(x_j)$ and $B_-(x_j)$ since the vertices of both sets are connected by a~complete bipartite graph.
  Suppose now that $S$ contains both sets $B_+(x_j)$ and $B_-(x_j)$.
  Since the variable $x_j$ appears at most $2p + 1$ times in $\varphi$, we can pick a~literal $\ell_j \in \{x_j, \neg x_j\}$ that appears at most $p$ times in $\varphi$.
  Consider now the following set $S^\star$:
  \[ S^\star = (S \setminus B(\ell_j)) \cup \{ a_i(s_1, \ldots, s_k) \in A \,\mid\, \ell_j \in \{s_1, \ldots, s_k\} \}. \]
  Observing that $S^\star$ is formed from $S$ by removing $B(\ell_j)$ and introducing all neighbors of $B(\ell_j)$ in $G(\varphi, p)$, we conclude that $S^\star$ is also a~vertex cover of $G(\varphi, p)$.
  Since $|B(\ell_j)| = \gamma$, $B(\ell_j) \subseteq S$ and $\ell_j$ appears at most $p$ times in $\varphi$, the cardinality of $S^\star$ is at most $|S^\star| \leq |S| - \gamma + 2^{k-1} p = |S|$.
  This contradicts the minimality of $S$.
  Therefore, for every variable $x_j$, $S$ contains fully one of the sets $B_+(x_j)$, $B_-(x_j)$ and is disjoint from the other.
  We obtain a~valuation $F$ of $\varphi$ by including in $F$, for every variable $x_j$, the literal $\ell_j \in \{x_j, \neg x_j\}$ such that $B(\ell_j) \subseteq S$.
  We aim to show that $F$ satisfies at least $m''$ clauses of $\varphi$.

  Note that $|S \cap B| = \gamma n$ and $|A| = (2^k - 1)m$ and so $|A \setminus S| \geq m''$.
  Moreover, $|A(C_i) \setminus S| \leq 1$ for each clause $C_i$ of $\varphi$; otherwise, we would have $a_i(s_1, \ldots, s_k), a_i(s'_1, \ldots, s'_k) \in A(C_i) \setminus S$, where $s'_j = \neg s_j$ for some $j \in [k]$.
  But then $s_j$ is connected to the vertices of $B(s_j)$ and $s'_j$ is connected to the vertices of $B(s'_j)$, and by the minimality of $S$, either of the sets $B(s_j)$, $B(s'_j)$ is disjoint from $S$; a~contradiction with the assumption that $S$ is a~vertex cover.
  Therefore, there exist at least $m''$ clauses $C_i$ of $\varphi$ such that $|A(C_i) \setminus S| = 1$.
  Observe that each such clause $C_i$ is satisfied by $F$.
  Indeed, suppose $C_i = \ell_1 \vee \ldots \vee \ell_k$ and let $a_i(s_1, \ldots, s_k)$ be the only element of $A(C_i) \setminus S$.
  Then $s_j \in \{\ell_j, \neg \ell_j\}$ for each $j \in [k]$ and $(s_1, \ldots, s_k) \neq (\neg \ell_1, \ldots, \neg \ell_k)$ by the construction of $G(\varphi, p)$, and $B(s_1) \cup \ldots \cup B(s_k) \subseteq S$ by the fact that $S$ is a~vertex cover.
  Thus $s_1, \ldots, s_k \in F$ by the construction of $F$ and so $F$ satisfies $\varphi$ (i.e., there is a~literal $s_j \in \{s_1, \ldots, s_k\}$ such that $s_j = \ell_j$).
\end{proof}

As an~immediate corollary, we get that:
\begin{corollary}
  \label{cor:formula-to-td}
  For all pairs of integers $k \geq 2$, $p \geq 1$ there exists a~polynomial-time algorithm that inputs a~$k$-CNF formula $\varphi$ with $n$ variables and $m$ clauses where every variable appears at most $2p + 1$ times,
  and outputs a~graph $H(\varphi,p)$ with $|V(H(\varphi,p))| = O(n)$ and the following property:
  If $m'$ is the maximum number of clauses that can be satisfied in $\varphi$, then
  \[
    \td(H(\varphi,p)) = 2(2^k - 1)m - m' + 2^k pn + 1.
  \]
\end{corollary}
\begin{proof}
  Apply \cref{lem:vc-to-td} to the graph $G(\varphi, p)$ with vertex set $A \cup B_+ \cup B_-$ and the parameter $\ell = |A \cup B_+| = (2^k - 1)m + 2^{k-1}pn$ (the choice of $\ell$ comes from the fact that $A \cup B_+$ is a~vertex cover of $G(\varphi, p)$).
  The value of $\td(H(\varphi,p))$ follows from \cref{lem:vc-to-td,lem:formula-to-vc}.
\end{proof}

We are now almost ready to show the approximation hardness of {\sc Treedepth} (\cref{thm:td-no-ptas}).
We start from the following approximation hardness of the maximization variant of \textsc{2-SAT}:

\begin{theorem}[{\cite[Theorem 12]{Berman03b}}]
  \label{thm:2cnf-hardness}
  For every $\varepsilon > 0$, within the family of $m$-clause $2$-CNF formulas where each variable appears $3$ or $4$ times, it is \textsf{NP}-hard to distinguish between the formulas where at least $(1 - \varepsilon)m$ clauses are satisfiable and formulas where at most $\left(\frac{251}{252} + \varepsilon\right)m$ clauses are satisfiable.
\end{theorem}

\begin{proof}[Proof of~\cref{thm:td-no-ptas}]
  Consider the family of $2$-CNF formulas where each variable appears $3$ or $4$ times.
  In this family, every $n$-variable, $m$-clause formula satisfies $2m \geq 3n$, or equivalently $n \leq \frac23 m$.
  By \cref{cor:formula-to-td}, every such formula $\varphi$ can be translated---in polynomial time---to a~graph $H(\varphi)$ with the property that if $m'$ is the maximum number of clauses satisfiable in $\varphi$, then $\td(H(\varphi)) = 6m - m' + 8n + 1$.

  Now let $\delta < \frac{1}{2604}$ be fixed.
  Then there is some $\varepsilon > 0$ such that, for large enough $m$ and $n \leq \frac23 m$,
  \[ (6m - (1 - \varepsilon)m + 8n + 1)(1 + \delta) < 6m - \left(\frac{251}{252} + \varepsilon\right)m + 8n + 1. \]
  So, letting $k = 6m - (1 - \varepsilon)m + 8n$, distinguishing between formulas $\varphi$ where at least $(1 - \varepsilon)m$ clauses are satisfiable and formulas where at most $\left(\frac{251}{252} + \varepsilon\right)m$ clauses are satisfiable reduces to distinguishing between graphs of treedepth at most $k$ and those of treedepth at least $(1 + \delta)k$.
  Hence, we obtain hardness by \cref{thm:2cnf-hardness}.
\end{proof}

We move to the hardness results under the Exponential Time Hypothesis (ETH); both presented proofs invoke the Sparsification Lemma of Impagliazzo, Paturi, and Zane~\cite{sparsification}.
We use this lemma in the following form:

\begin{lemma}[Sparsification Lemma {\cite{sparsification}}]
  \label{lem:sparsification}
  For every $0 < \varepsilon' < 0.1$ there exists a~constant $B > 0$ such that a~$3$-CNF formula $\varphi$ with $n'$ variables can be transformed in time $O(2^{\varepsilon' n'})$ into $s \leq 2^{\varepsilon' n'}$ $3$-CNF formulas $\varphi_1, \ldots, \varphi_s$ with the same set of variables such that
  \begin{enumerate}[label = (\roman*)]
   \item each variable appears at most $B$ times in each formula $\varphi_i$, and
   \item $\varphi$ is satisfiable if and only if one of the formulas $\varphi_i$ is satisfiable.
  \end{enumerate}
\end{lemma}

The Sparsification Lemma, when combined with ETH and \cref{cor:formula-to-td}, yields a~straightforward proof of \cref{thm:td-eth-lower-bound}:

\begin{proof}[Proof of \cref{thm:td-eth-lower-bound}]
  \cref{lem:sparsification} together with ETH implies that, for some absolute constants $\varepsilon', B > 0$, no algorithm solves $3$-CNF instances of the satisfiability problem on $n'$ variables, with each variable appearing at most $B$ times, in time $O(2^{\varepsilon' n'})$.
  Given such an~instance $\varphi$ with $n'$ variables and $m'$ clauses, we use \cref{cor:formula-to-td} with $k = 3$, $p = \left\lfloor \frac{B}{2} \right\rfloor$ to transform it, in polynomial time, into a~graph $H(\varphi)$ with $|V(H(\varphi))| \leq c \cdot n'$ for some fixed constant $c$ (where $n'$ is sufficiently large).
  Then $\varphi$ is satisfiable if and only if $\td(H(\varphi)) = (2^{k+1} - 3)m' + 2^kpn'$.
  Therefore, under ETH, the treedepth of an~$n$-vertex graph cannot be computed in time $O(2^{\varepsilon n})$ where $\varepsilon \coloneqq \varepsilon'/(2c)$.
\end{proof}

The approximation time complexity lower bound (\cref{thm:td-eth-approx-lower-bound}) requires a~bit more work:

\begin{proof}[Proof of \cref{thm:td-eth-approx-lower-bound}]
  Assuming ETH, there is some $\lambda > 0$ such that $n'$-variable $3$-SAT cannot be solved in time $O(2^{\lambda n'})$.
  Pick $0 < \varepsilon' < \min\{0.1, \frac{1}{2}\lambda\}$ and invoke the Sparsification Lemma (\cref{lem:sparsification}).
  Let the $3$-CNF formulas $\varphi$ and $\varphi_1, \ldots, \varphi_s$ be as in the statement of the lemma.
  Combining the Sparsification Lemma with a~polynomial-time \textsc{Satisfiability} inapproximability framework of H{\aa}stad~\cite{Hastad01} and a~quasi-linear-size construction of polynomially-checkable proofs (PCP) by Bafna, Minzer, Vyas, and Yun~\cite{BafnaMVY25}, we find absolute constants $0 < \alpha_1 < \alpha_2 < 1$ and $B^\star, c > 0$ such that we can translate each formula $\varphi_i$ in polynomial time to a~$3$-CNF formula $\varphi^\star_i$ with $n^\star \leq n' \log^c n'$ variables, each appearing at most $B^\star$ times, with the following property:
  If $\varphi^\star_i$ has $m^\star$ clauses, then:
  \begin{itemize}
    \item If $\varphi_i$ is satisfiable, then at least $\alpha_2 m^\star$ clauses of $\varphi^\star_i$ can be satisfied.
    \item If $\varphi_i$ is not satisfiable, then at most $\alpha_1 m^\star$ clauses of $\varphi^\star_i$ can be satisfied.
  \end{itemize}

  (This reduction is standard; see \cite[Lemma 2]{Bonnet15} for the proof of an~analogous result under the assumption of the existence of a~linear-size PCP construction.)
  Chaining this result with \cref{cor:formula-to-td}, we find absolute constants $0 < \beta_1 < \beta_2 < 1$ and a~polynomial-time algorithm that transforms each $\varphi^\star_i$ into a~graph $H_i$ with at most $dn^\star$ vertices, where $d$ is a fixed constant, such that:
  \begin{itemize}
    \item If $\varphi_i$ is satisfiable, then $\td(H_i) \leq \beta_1 |V(H_i)|$.
    \item If $\varphi_i$ is not satisfiable, then $\td(H_i) > \beta_2 |V(H_i)|$.
  \end{itemize}

  Now, let $\varepsilon \coloneqq \varepsilon'/d$ and $\delta = \beta_2 / \beta_1 - 1$.
  Then, supposing that a~$O(2^{\varepsilon n / \log^c n})$-time $(1+\delta)$-approximation algorithm for treedepth existed, the satisfiability of $\varphi$ could be decided in time
  \[
    O(2^{\varepsilon' n'}) + s \cdot ((n')^{O(1)} + O(2^{\varepsilon dn^\star / \log^c n^\star})) \leq O(2^{\varepsilon' n'}) + 2^{\varepsilon' n'} \cdot ((n')^{O(1)} + O(2^{\varepsilon' n'})) = O(2^{\lambda n'}),
  \]
  thus refuting the ETH.
\end{proof}

\bibliographystyle{plainurl}
\bibliography{main}

\end{document}